\documentclass{elsarticle}
\usepackage[noend]{algpseudocode}

\errorcontextlines\maxdimen
\makeatletter
\newcommand*{\algrule}[1][\algorithmicindent]{\makebox[#1][l]{\hspace*{0.3em}\vrule height .75\baselineskip depth .25\baselineskip}}%

\newcount\ALG@printindent@tempcnta
\def\ALG@printindent{%
		\ifnum \theALG@nested>0
		\ifx\ALG@text\ALG@x@notext
		\addvspace{-1.35pt}
		\else
		\unskip
		\ALG@printindent@tempcnta=1
		\loop
		\algrule[\csname ALG@ind@\the\ALG@printindent@tempcnta\endcsname]%
		\advance \ALG@printindent@tempcnta 1
		\ifnum \ALG@printindent@tempcnta<\numexpr\theALG@nested+1\relax
		\repeat
		\fi
		\fi
}%

\usepackage{etoolbox}

\patchcmd{\ALG@doentity}{\noindent\hskip\ALG@tlm}{\ALG@printindent}{}{\errmessage{failed to patch}}

\makeatother

\usepackage{algorithm}
\usepackage{textcomp}
\usepackage{lineno,hyperref}
\usepackage{subfig}
\usepackage{graphicx}
\usepackage{amsmath}
\usepackage{amsthm}
\usepackage{amssymb}
\usepackage[utf8]{inputenc}
\usepackage{todonotes}
\usepackage{array}
\usepackage{booktabs,ragged2e}

\newcolumntype{P}[1]{>{\centering\arraybackslash}p{#1}}
\newcolumntype{M}[1]{>{\centering\arraybackslash}p{#1}}

\renewcommand{\algorithmicrequire}{\textbf{Input: A link stream $L$ defined over a set of nodes $V$, a duration $\Delta$}}
\renewcommand{\algorithmicensure}{\textbf{Output: The set of all maximal \dclique{}s in $L$}}
\newtheorem{lemma}{Lemma}
\newtheorem{theorem}{Theorem}

\modulolinenumbers[5]

\newcommand{\dclique}{$\Delta$-clique}
\newcommand{\clique}[3]{\ensuremath{(#1,[#2,#3])}}

\newcommand{\ie}{i.e.}

\author{Tiphaine Viard, Matthieu Latapy, Cl\'{e}mence Magnien}

\address{Sorbonne Universités, UPMC Univ Paris 06, 
	CNRS, LIP6 UMR 7606, \\
	4 place Jussieu 75005 Paris}

\ead{first.last@lip6.fr}

\journal{Theoretical Computer Science}

\begin{document}

\begin{frontmatter}

\title{Computing maximal cliques in link streams}

\begin{abstract}

A link stream is a collection of triplets $(t,u,v)$ indicating that an interaction occurred between $u$ and $v$ at time $t$. We generalize the classical notion of cliques in graphs to such link streams: for a given $\Delta$, a $\Delta$-clique is a set of nodes and a time interval such that all pairs of nodes in this set interact at least once during each sub-interval of duration $\Delta$. We propose an algorithm to enumerate all maximal (in terms of nodes or time interval) cliques of a link stream, and illustrate its practical relevance on a real-world contact trace.

\end{abstract}


\end{frontmatter}

\section{Introduction}
\label{introduction}

In a graph $G=(V,E)$ with $E \subseteq V\times V$, a clique $C\subseteq V$ is a set of nodes such that $C\times C \subseteq E$. In addition, $C$ is maximal if it is included in no other clique. In other words, a maximal clique is a set of nodes such that all possible links exist between them, and there is no other node linked to all of them. Enumerating maximal cliques of a graph is one of the most fundamental problems in computer science, and it has many applications~\cite{Rowe2007,Samudrala1998}.

A link stream $L = (T,V,E)$ with $T = [\alpha,\omega]$ and $E\subseteq T\times V\times V$ models interactions over time: $l = (t,u,v)$ in $E$ means that an interaction occurred between $u \in V$ and $v \in V$ at time $t \in T$. Link streams, also called temporal networks or time-varying graphs depending on the context, model many real-world data like contacts between individuals, email exchanges, or network traffic~\cite{Casteigts2011,Holme2011,Viard2014,Wehmuth2014}.

For a given duration $\Delta$, a \dclique\ $C$ of $L$ is a pair $C = \clique{X}{b}{e}$ with $X \subseteq V$ and $[b,e] \subseteq T$ such that $|X|\ge 2$,
and for all $\{u, v\} \subseteq X$ and $\tau \in [b,\max(e-\Delta,b)]$ there is a link $(t,u,v)$ in $E$ with $t \in [\tau,\min(\tau+\Delta,e)]$.
Notice that \dclique{}s necessarily have at least two nodes.


More intuitively, all nodes in $X$ interact at least once with all others at least every $\Delta$ from time $b$ to time $e$. 
\dclique{} $C$ is maximal if it is included in no other \dclique{},
(\ie{} there exists no \dclique{} $C' = \clique{X'}{b'}{e'}$ such that $C'\ne C$, $X \subseteq X'$ and $[b,e]\subseteq [b',e']$). See Figure~\ref{fig:example} for an example.

\begin{figure}

\centering

\includegraphics[width=.9\linewidth]{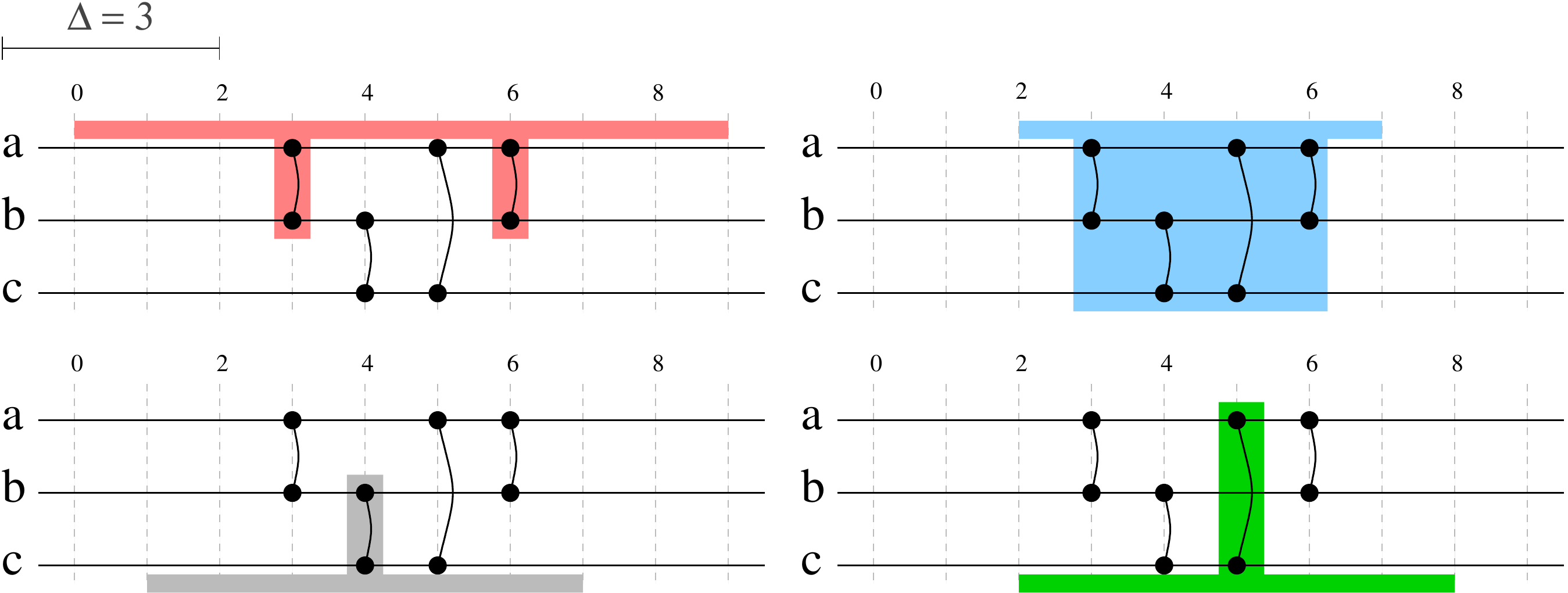}

\caption{Examples of \dclique s. We consider the link stream $L = ([0,9],\{a,b,c\},E)$ with $E = \{(3,a,b), (4,b,c), (5,a,c), (6,a,b)\}$ and $\Delta = 3$. There are four maximal $3$-cliques in $L$: $(\{a,b\},[0,9])$ (top left), $(\{a,b,c\}, [2,7])$ (top right), $(\{b,c\}, [1,7])$ (bottom left), and $(\{a,c\}, [2,8])$ (bottom right). Notice that $(\{a,b,c\},[1,7])$ is not a \dclique\ since during time interval $[1,4]$ of duration $\Delta=3$ there is no interaction between $a$ and $c$. Notice also that $(\{a,b\},[1,9])$, for instance, is not maximal: it is included in $(\{a,b\},[0,9])$.}

\label{fig:example}
\end{figure}

In real-world situations like the ones cited above, $\Delta$-cliques are signatures of meetings, discussions, or distributed applications for instance. Moreover, just like cliques in a graph correspond to its subgraphs of density $1$, $\Delta$-cliques in a link stream correspond to its substreams of $\Delta$-density $1$, as defined in~\cite{Viard2014}. Therefore, $\Delta$-cliques in link streams are natural generalizations of cliques in graphs.

In this paper, we propose the first algorithm for listing all maximal \dclique s of a given link stream. We illustrate the relevance of the concept and algorithm by computing maximal \dclique s of a real-world dataset.

Before entering in the core of the presentation, notice that we consider here undirected links only: given a link stream $L = (T,V,E)$, we make no distinction between $(t,u,v) \in E$ and $(t,v,u) \in E$. Likewise, we suppose that there is no loop $(t,v,v)$ in $E$, and no isolated node ($\forall v\in V,\ \exists (t,u,v)\in E$).

We finally define the first occurrence time of $(u,v)$ after $b$ as the smallest time $t \ge b$ such that $(t,u,v)\in E$, and we denote it by $f_{buv}$. Conversely we denote the last occurrence time of $(u,v)$ before $e$ by $l_{euv}$. We say that a link $(t,u,v)$ is in $C = \clique{X}{b}{e}$ if $u \in X$, $v \in X$ and $t \in [b,e]$.

\section{Algorithm}

\label{sec:algorithm}

One may trivially enumerate all maximal cliques in a graph as follows.
One maintains a set $M$ of previously found cliques (maximal or not),
as well as a set $S$ of candidate cliques.
Then for each clique $C$ in $S$, one removes $C$ from $S$ and searches for nodes outside $C$ connected to all nodes in clique $C$, thus obtaining new cliques (one for each such node) larger than $C$. If one finds no such node, then clique $C$ is maximal and it is part of the output.
Otherwise, if the newly found cliques have  not already been found
(\ie{}, they do not belong to $M$),
then one adds them  to $S$ and $M$.
The set $S$ is initialized with the trivial cliques containing only one node, and all maximal cliques have been found when $S$ is empty.
The set $M$ is used for memorization, 
and ensures that one does not examine the same clique more than once.
In \cite{Johnson1988} the authors use this framework to enumerate all maximal cliques of a graph in lexicographic order.

Our algorithm for finding \dclique s in link stream $L = (T,V,E)$ (Algorithm~\ref{alg:dcliques}) relies on the same scheme.
We initialize the set $S$ of candidate \dclique{}s and the set $M$ of all found \dclique{}s
with the trivial \dclique s $\clique{\{a,b\}}{t}{t}$ for all $(t,a,b)$ in $E$ (Line~\ref{alg:init_state}).
Then, until $S$ is empty ({\em while} loop of Lines~\ref{alg:begin_loop} to~\ref{alg:end_loop}),
we pick an element $\clique{X}{b}{e}$ in $S$ (Line~\ref{alg:get_clique}) and search for nodes $v$ outside $X$ such that $\clique{X\cup\{v\}}{b}{e}$\ is a \dclique\ (Lines~\ref{alg:add_node_begin} to~\ref{alg:add_node_end}).
We also look for a value $b'<b$ such that $\clique{X}{b'}{e}$ is a \dclique\ (Lines~\ref{alg:get_f} to~\ref{alg:add_clique_time_b}), and likewise a value $e'>e$ such that $\clique{X}{b}{e'}$ is a \dclique\ (Lines~\ref{alg:get_l} to~\ref{alg:add_clique_time_e}).
If we find such a node, such a $b'$ or such an $e'$,
then \dclique{} $C$ is not maximal and we add to $S$ and $M$ the new \dclique{}s larger than $C$ we just found
(Lines~\ref{alg:add_clique_node},~\ref{alg:add_clique_time_b} and~\ref{alg:add_clique_time_e}),
on the condition that they had not already been seen (\ie{}, they do not belong to $M$).
Otherwise, $C$ is maximal and is part of the output (Line~\ref{alg:add_c_r}).

\begin{algorithm}
\caption{Maximal \dclique s of a link stream}
\label{alg:dcliques}
\renewcommand{\algorithmicrequire}{\textbf{input:} a link stream $L = (T,V,E)$ and a duration $\Delta$}
\renewcommand{\algorithmicensure}{\textbf{output:} the set of all maximal \dclique{}s in $L$}
\algorithmicrequire\\
\algorithmicensure

\begin{algorithmic}[1]
\State $S \gets \emptyset$, $R \gets \emptyset$, $M \gets \emptyset{}$
\State for $(t,u,v) \in E$: add $\clique{\{u,v\}}{t}{t}$ to $S$ and to $M$ \label{alg:init_state}
\While{$S \ne \emptyset$} \label{alg:begin_loop}
 \State take and remove $\clique{X}{b}{e}$ from $S$ \label{alg:get_clique}
 \State set isMax to True
 \For{$v$ in $V\setminus X$}\label{alg:add_node_begin}
  \If{$\clique{X\cup \{v\}}{b}{e}$ is a $\Delta$-clique}\label{alg:check_clique_node}
   \State set isMax to False \label{alg:nodeismaxfalse}
   \If{$\clique{X\cup \{v\}}{b}{e}$ not in $M$} \label{alg:ifXvinM}
    \State add $\clique{X\cup \{v\}}{b}{e}$ to $S$ and $M$ \label{alg:add_clique_node} 
   \EndIf
  \EndIf
 \EndFor \label{alg:add_node_end}
 \State $f \gets \max_{u,v\in X} f_{buv}$ \label{alg:get_f} \Comment{{\em\footnotesize latest first occurrence time of a link in \clique{X}{b}{e}}}
 \State set $b'$ to $f - \Delta$\label{alg:b_default}
 \If{$b \ne b'$} \label{alg:check_b_f}
  \State set isMax to False \label{bismaxfalse}
  \If{$\clique{X}{b'}{e}$ not in $M$}~\label{alg:ifbpinM}
    \State add $\clique{X}{b'}{e}$ to $S$ and $M$ \label{alg:add_clique_time_b}
  \EndIf
 \EndIf
 \State $l \gets \min_{u,v\in X} l_{euv}$ \label{alg:get_l} \Comment{{\em\footnotesize earliest last occurrence time of a link in \clique{X}{b}{e}}}
 \State set $e'$ to $l + \Delta$
 \If{$e \ne e'$}\label{alg:check_e_l}
   \State set isMax to False \label{alg:eismaxfalse}
   \If{$\clique{X}{b}{e'}$ not in $M$}
     \State add $\clique{X}{b}{e'}$ to $S$ and $M$ \label{alg:add_clique_time_e}
   \EndIf
 \EndIf
 \If{isMax}\label{alg:test_max}
  \State add $\clique{X}{b}{e}$ to $R$\label{alg:add_c_r}
 \EndIf
\EndWhile\label{alg:end_loop}
\State \Return $R$
\end{algorithmic}
\end{algorithm}

\label{algorithm}

Let us explain the choice of $b'$ (Lines~\ref{alg:get_f} to~\ref{alg:add_clique_time_b}) in details, the choice of $e'$ (Lines~\ref{alg:check_e_l} to~\ref{alg:add_clique_time_e}) being symmetrical. 
For a given \dclique{} $\clique{X}{b}{e}$,
we set $b'$ to $f-\Delta$, which is {the smallest time} such that we are sure that $\clique{X}{b'}{e}$ is a \dclique{}
without inspecting any link outside of \clique{X}{b}{e}.
Indeed, 
all links in $X\times X$ appear at least once in the interval $[f-\Delta, f]$:
$f$ is the latest of the first occurrence times of all links in this \dclique{}, and
so all links appear at least once in $[b,f] \subseteq [f-\Delta, f]$. 
If $b' \not = b$, then the \dclique{} \clique{X}{b'}{e} is added to $S$ (Line~\ref{alg:check_b_f}).

We display in Figure~\ref{fig:procedure} an example of a sequence of such operations from an initial trivial \dclique{} to a maximal \dclique{} in an illustrative link stream.
The algorithm builds this way a set of \dclique s of $L$, which we call the {\em configuration space};
we display the configuration space for this simple example in Figure~\ref{fig:configspace} together with the relations induced by the algorithm between these \dclique s.

\begin{figure}[htbp]

\centering

\includegraphics[width=0.9\linewidth]{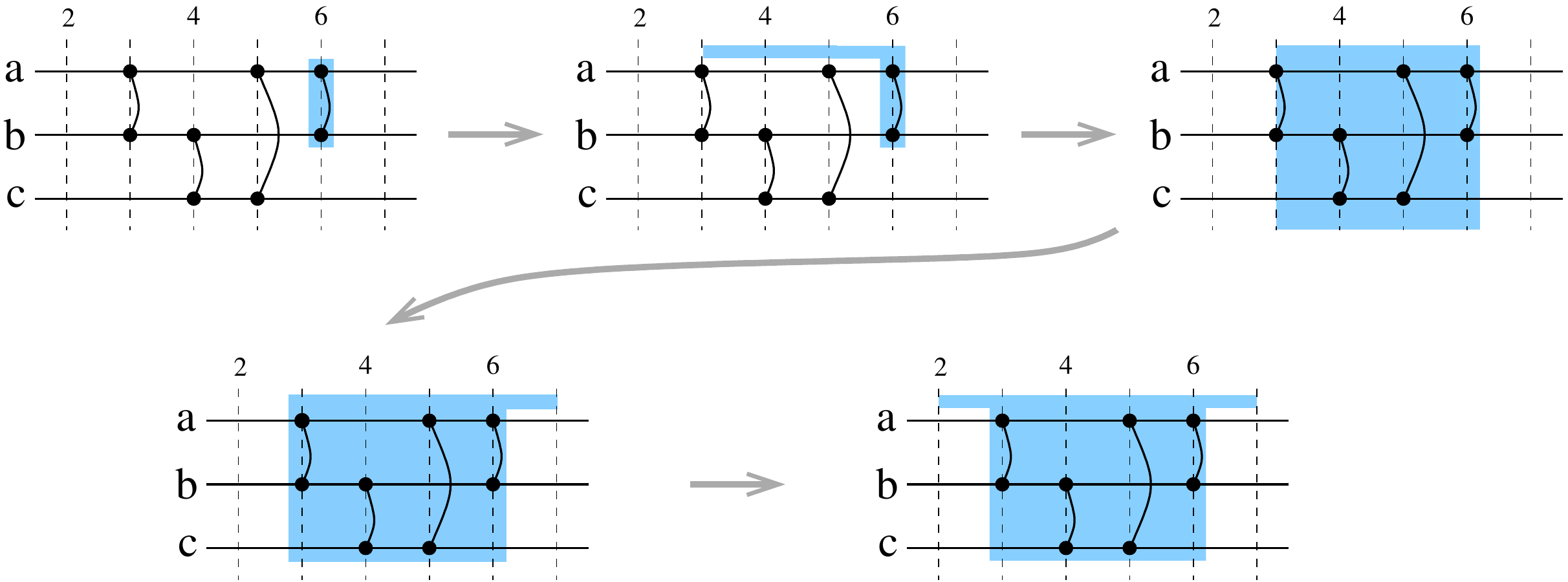}

\caption{A sequence of \dclique s built by our algorithm to find a maximal \dclique\ (bottom-right) from an initial trivial \dclique\ (top-left) in the link stream of Figure~\ref{fig:example} when $\Delta=3$.\\
From left to right and top to bottom: the algorithm starts with $\clique{\{a,b\}}{6}{6}$, and finds $\clique{\{a,b\}}{3}{6}$ thanks to Lines~\ref{alg:get_f} to~\ref{alg:add_clique_time_b} of the algorithm.
It then finds $\clique{\{a,b,c\}}{3}{6}$ thanks to Lines~\ref{alg:add_node_begin} to~\ref{alg:add_clique_node}.
It finds $\clique{\{a,b,c\}}{3}{7}$ from Lines~\ref{alg:get_l} to~\ref{alg:add_clique_time_e},
and finally $\clique{\{a,b,c\}}{2}{7}$ from Lines~\ref{alg:get_f} to~\ref{alg:add_clique_time_b}.}

\label{fig:procedure}

\end{figure}

\begin{figure}[!t]

  \centering

  \includegraphics[width=0.9\linewidth]{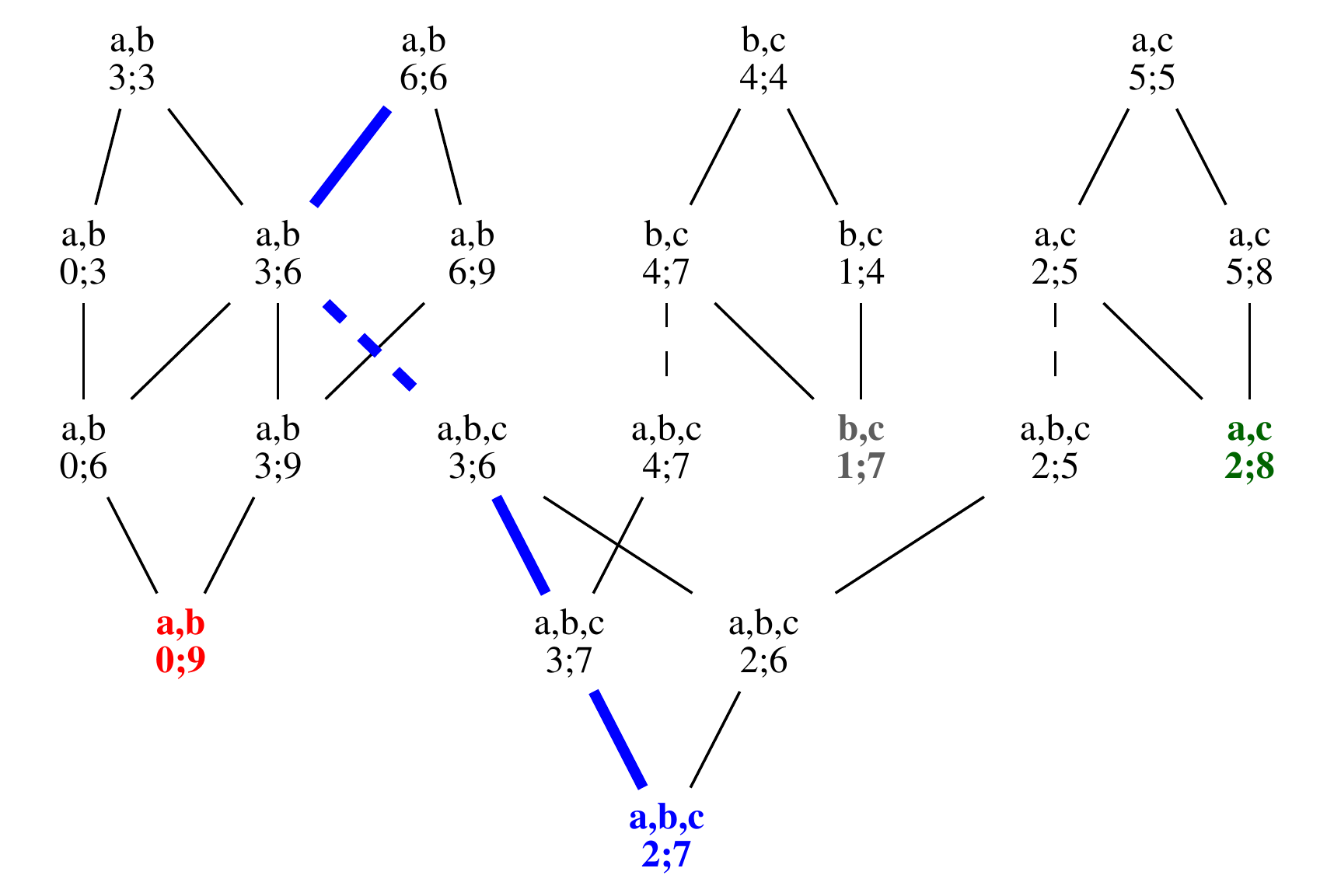}

  \caption{The configuration space built by our algorithm from the link stream of Figures~\ref{fig:example} and~\ref{fig:procedure} when $\Delta=3$. Each element is a \dclique\ and it is linked to the \dclique s the algorithm builds from it (links are implicitly directed from top to bottom).
Plain links indicate \dclique s discovered by Lines~\ref{alg:get_f} to~\ref{alg:add_clique_time_b}
or Lines~\ref{alg:get_l} to~\ref{alg:add_clique_time_e} of the algorithm,
which change the time span of the clique. Dotted links indicate \dclique s discovered by Lines~\ref{alg:add_node_begin} to~\ref{alg:add_node_end}, which change the set of nodes involved in the clique. The bold path is the one detailed in Figure~\ref{fig:procedure}. Colors correspond to the maximal \dclique s displayed in Figure~\ref{fig:example}.}

  \label{fig:configspace}

\end{figure}

To prove the validity of Algorithm~\ref{alg:dcliques},
we must show that all the elements it outputs are \dclique{}s, that they are maximal,
and that all maximal \dclique{}s are in its output.

\begin{lemma}[]\label{lemma:clique_transfo}
In Algorithm~\ref{alg:dcliques}, all elements of S are \dclique s of $L$.
\label{lemma1}
\end{lemma}

\begin{proof}

We prove the claim by induction on the iterations of the {\em while} loop (Lines~\ref{alg:begin_loop} to~\ref{alg:end_loop}).

Initially, all elements of $S$ are \dclique s (Line~\ref{alg:init_state}).
Let us assume that all the elements of $S$ are \dclique s at the $i$-th iteration of the loop (induction hypothesis). The loop may add new elements to $S$ at Lines~\ref{alg:add_clique_node}, \ref{alg:add_clique_time_b} and~\ref{alg:add_clique_time_e}.
In all cases, the added element is built from an element $C=\clique{X}{b}{e}$ of $S$ (Line~\ref{alg:get_clique}), which is a \dclique\ by induction hypothesis.

It is trivial (from the test at Line~\ref{alg:check_clique_node}) that Line~\ref{alg:add_clique_node} only adds \dclique s.

Let us show that $(X,[b,l+\Delta])$, 
where $l$ is computed in Line~\ref{alg:get_l}, necessarily is a \dclique. 
As $\clique{X}{b}{e}$ is a \dclique\ all links in $X\times X$ appear at least once every $\Delta$ from $b$ to $l\le e$.
Moreover, since $l$ is the earliest last occurrence time of a link in $C$,
for all $u$ and $v$ in $X$ there is necessarily a link $(t,u,v)$ in $E$ with $l\le t\le e$.
Notice also that $l \ge e-\Delta$, otherwise $\clique{X}{b}{e}$ would not be a \dclique. 
Therefore a link between $u$ and $v$ occurs at least once between 
$l$ and $l+\Delta$
for all $u$ and $v$ in $X$.
Finally, $\clique{X}{b}{l+\Delta}$ is a \dclique{}.

The same arguments hold for Line~\ref{alg:get_f}.

Finally, at the end of the $(i+1)$-th iteration of the loop, all the elements of $S$ are \dclique s, which ends the proof.
\end{proof}

\begin{lemma}\label{lemma:maximal}
All the elements of the set returned by Algorithm~\ref{alg:dcliques} are maximal \dclique s of $L$.
\end{lemma}

\begin{proof}

Let $C=\clique{X}{b}{e}$ be an element of $R$ returned by the algorithm.
Only elements of $S$ are added to $R$ (at Line~\ref{alg:add_c_r}), and so according to Lemma~\ref{lemma1} $C$ is a \dclique.
Assume it is not maximal; then we are in one of the three following situations.

There exists $v$ in $V\setminus X$ such that $\clique{X\cup \{v\}}{b}{e}$ is a \dclique. Then $v$ is found at Lines~\ref{alg:add_node_begin}\textendash\ref{alg:check_clique_node},
and Line~\ref{alg:nodeismaxfalse} sets the boolean {\em isMax} to {\em false}. Therefore, Line~\ref{alg:test_max} ensures that $C=\clique{X}{b}{e}$ is not added to $R$, and we reach a contradiction.

There exists $e'>e$ such that $\clique{X}{b}{e'}$ is a \dclique\ and we assume without loss of generality that there is no link between nodes in $X$ from $e$ to $e'$.
Then, let us consider  $l \in [b,e]$, computed in Line~\ref{alg:get_l}, which
is the earliest last occurrence time of a link in $C$.
We necessarily have $l \ge e'-\Delta$ because  $\clique{X}{b}{e'}$ is a \dclique.
Since  $e'>e$, this implies
 $l > e-\Delta$.
As a consequence, the test in Line~\ref{alg:check_e_l} of the algorithm is satisfied, and Line~\ref{alg:eismaxfalse} sets the boolean {\em isMax} to {\em false}. Like above, we reach a contradiction.

If there exists $b'<b$ such that $\clique{X}{b'}{e}$ is a \dclique, then similarly to the previous case we reach a contradiction.

Finally, $C$ necessarily is maximal, which proves the claim.
\end{proof}

\newcommand{\ff}{{\ensuremath{s}}}

Before proving our main result, which is that all maximal \dclique{}s are returned
by the algorithm,
we need the following two intermediate results.

\begin{lemma}
  \label{lem:width}
  Let $C = \clique{X}{b}{e}$ be a maximal \dclique{} of $L$,
  and let $\ff$ be the earliest occurrence time of a link in $C$.
  Then $e \ge \ff{} + \Delta$.
\end{lemma}

\begin{proof}
  Since $C$ is a \dclique{} and by definition of \ff{}, for all $u, v$ in $X$ there exists at least one link $(t,u,v)$ such that $\ff{} \le t\le e$.  Assume $e < \ff + \Delta$; then for all $u,v$ in $X$ there also exists a link $(t,u,v)$ such that $\ff{} \le t \le e < \ff + \Delta$.  Therefore $\clique{X}{b}{\ff+\Delta}$ is a \dclique{} and $C$ is included in it, which means that $C$ is not maximal and we reach a contradiction.
\end{proof}

\begin{lemma}
\label{lem:path}
Let $C = \clique{X}{b}{e}$ be a maximal \dclique{} of $L$ and let $\ff$ be the earliest occurrence time of a link in $C$.
If 
$\clique{X}{\ff}{\ff+\Delta}$ is in $S$ at some stage of Algorithm~\ref{alg:dcliques}, then $C$ is in the set returned by the algorithm.
\end{lemma}

\begin{proof}
Assume $C_0 = \clique{X}{\ff}{\ff+\Delta}$ is in $S$ and consider the longest sequence of steps of Algorithm~\ref{alg:dcliques} of the form:
$C_0 \rightarrow C_1 \rightarrow \cdots \rightarrow C_k$ such that for all $i$ $C_i = \clique{X}{\ff}{e_i}$ with $e_{i+1}>e_i$. In other words, the algorithm builds $C_{i+1}$ from $C_i$ in Lines~\ref{alg:get_l} to~\ref{alg:add_clique_time_e}
(notice that $e\ge s+\Delta$ from Lemma~\ref{lem:width} and so $C_0$ is included in $C$).

We prove that $C_k = \clique{X}{\ff}{e}$ by contradiction. Assume this is false, and so that $e_k \not= e$. 
As $C$ is maximal, we then necessarily have $e_k < e$.
In addition, $e_k = l+\Delta$ where $l$ is the earliest last occurrence time of a link in $C_{k-1}$ computed at Line~\ref{alg:get_l}.
Since $C_k$ is the last \dclique{} in the sequence, 
$l$ is also the earliest last occurrence time of a link in $C_k$
(otherwise there would be a clique $C_{k+1}$ satisfying the constraints of the sequence above).
Therefore there exist $u,v\in X$ such that $(l,u,v)\in E$ and such that
there is no occurrence of a link $(u,v)$ between $l$ and $e_k = l+\Delta$.
This ensures that there exists an $\epsilon$ such that 
$l + \Delta + \epsilon < e$ and such that there is no link between $u$ and $v$ from $l+ \epsilon$ to $l + \Delta + \epsilon$, which contradicts the assumption that $C$ is a \dclique{}.

We now show that the algorithm builds $C$ from $C_k$ to end the proof. Since $C$ is maximal, there exists $u,v \in X$ such that $(b+\Delta, u,v) \in E$ and such that there is no other link between $u$ and $v$ from $b$ to $b+\Delta$. By definition of $\ff$, $b+\Delta \ge \ff$. Therefore the latest first occurrence time of a link in $C_k$, $f$, is equal to $b+\Delta$ and Lines~\ref{alg:get_f} to~\ref{alg:add_clique_time_b} build $C$ from $C_k$.
\end{proof}

\begin{lemma}[]\label{lemma:clique_decouverte}
All maximal \dclique s of $L$ are in the set returned by Algorithm~\ref{alg:dcliques}. 
\end{lemma}

\begin{proof}
It is easy to check that if $S$ contains a maximal \dclique\ then it is added to the set $R$ returned by the algorithm, and only these \dclique{}s are added to $R$.
We therefore show that all maximal \dclique s are in $S$ at some stage. 

Let $C = \clique{X}{b}{e}$ be a maximal \dclique{} of $L$, let $\ff$ be the earliest occurrence time of a link in $C$, and let $u,v\in X$ be two nodes such that there exists a link between them at $\ff$ (\ie{}, $(\ff, u,v) \in E$).
We show that there is a sequence of steps of the algorithm that builds $C$ from \dclique{} $C_0 = \clique{\{u,v\}}{\ff}{\ff}$ (which is placed in $S$ at the beginning of the algorithm, Line~\ref{alg:init_state}).

Lines~\ref{alg:get_l} to~\ref{alg:add_clique_time_e} builds $C_1 = \clique{\{u,v\}}{\ff}{\ff+\Delta}$ from $C_0$.

Notice that for all subset $Y$ of $X$, $\clique{Y}{\ff}{\ff+\Delta}$ is a \dclique{}. Therefore the algorithm iteratively adds all elements of $X$ at Lines~\ref{alg:add_node_begin} to~\ref{alg:add_node_end}, finally obtaining $C'=\clique{X}{\ff}{\ff+\Delta}$ from $C_1$.

We finally apply Lemma~\ref{lem:path} to conclude that the algorithm builds $C$ from $C'$.

\end{proof}

From these lemmas, we finally obtain the following result.

\begin{theorem}
Given a link stream $L$ and a duration $\Delta$, Algorithm~\ref{alg:dcliques} computes the set of all maximal \dclique s of $L$.
\end{theorem}

In order to investigate the complexity of our algorithm, let us denote by $n = |V|$ the number of nodes and $m = |E|$ the number of links in $L$. First notice that the number of elements in the configuration space built by the algorithm is bounded by the number of subsets of $V$ times the number of sub-intervals of $T$.

Moreover, for all \dclique{} $C=\clique{X}{b}{e}$ in the configuration space, there exists a link $(b, u, v)$ or a link $(b+\Delta, u, v)$ in $E$.
Indeed, the initial trivial \dclique{}s are in the first case,
and all \dclique{}s obtained from them are also in this case until Line~\ref{alg:add_clique_time_b} is applied.
The \dclique{}s built after this are in the second case. Likewise, there exists a link $(e,u,v)$ or a link $(e-\Delta, u, v)$ in $E$. Therefore, the number of possible values for $b$ and $e$ for any \dclique{} in the configuration space is proportional to the number of time instants at which a link occurs, which is bounded by the number of links $m$. The number of sub-intervals of $T$ corresponding to a \dclique{} in the configuration space is therefore in ${\cal O}(m^2)$. This bound is reached in the worst case, for instance if the stream is a sequence of links occurring once every $\Delta$ time interval.

The trivial bound ${\cal O}(2^n)$ for the number of subsets of $V$ is also reached in the worst case, for instance if there is a link between all pairs of nodes at the same time: the algorithm will enumerate all subsets of $V$.

Therefore, the number of elements in the configuration space is in ${\cal O}(2^n m^2)$. This leads to the space complexity of our algorithm: it is proportional to the space needed to store the configuration space, which is in ${\cal O}(2^n n m^2)$ since each element may be stored in ${\cal O}(n)$ space.

\medskip

We estimate the time complexity by studying the complexity of operations performed on each element of the configuration space, (\ie{}, the complexity of each iteration of the \emph{while} loop at Lines~\ref{alg:begin_loop} to~\ref{alg:end_loop}). Let us consider a \dclique{} $C = (X, [b,e])$ picked from $S$ by the algorithm at Line~\ref{alg:get_clique}.

The {\em while} loop is composed of three blocks:
(1) searching for \dclique{}s of the form $(X\cup\{v\}, [b,e])$ larger than $C$ (Lines~\ref{alg:add_node_begin} to~\ref{alg:add_node_end});
(2) searching for a \dclique{} $(X, [b',e])$ larger than $C$ (Lines~\ref{alg:get_f} to~\ref{alg:add_clique_time_b});
and (3) searching for a \dclique{} $(X, [b,e'])$ larger than $C$ (Lines~\ref{alg:get_l} to~\ref{alg:add_clique_time_e}). The third block has the same complexity as the second one, so we focus on the time complexity of the two first blocks.

Given a node $v \notin X$,
Line~\ref{alg:check_clique_node} tests whether for all nodes $u$ in $X$ there is a link $(t,v,u)\in E$ in each time interval of duration $\Delta$.
This requires at most $|X| \cdot m$ tests, and so it is in ${\cal O}(n m)$.
Then, Line~\ref{alg:ifXvinM} searches for the found \dclique{} in $M$, which has a size in ${\cal O}(2^n m^2)$.
Since the comparison between two \dclique{}s can be performed in ${\cal O}(n)$ time,
this search therefore is in ${\cal O}(n\log(2^n m^2)) = {\cal O}(n^2 + n\log m)$ time.
The algorithm repeats these operations for all nodes $v\in V\setminus X$, and thus less than $n$ times,
hence the complexity of Lines~\ref{alg:add_node_begin} to~\ref{alg:add_clique_node} is in ${\cal O}(n(nm + n^2 + n\log m)) = {\cal O}(n^2 m+n^3)$.

Computing $f$ in Line~\ref{alg:get_f} may clearly be done with at most $m$ tests. Lines~\ref{alg:check_b_f} and~\ref{bismaxfalse} are all trivial computations.
Lines~\ref{alg:ifbpinM} and Line~\ref{alg:add_clique_time_b} are in ${\cal O}(n^2 + n\log m)$.
The complexity of Lines~\ref{alg:get_f} to~\ref{alg:add_clique_time_b} is therefore in ${\cal O}(m + n^2 + n\log m)$.

Finally, each iteration of the while loop costs at most
${\cal O}(n^2 m +n^3 + m + n^2 +n\log m) = {\cal O}(n^2 m + n^3)$ time.
We bound the overall time complexity of the algorithm by multiplying this by the number of iterations of the while loop, which is the number of elements in the configuration space.
It is therefore in ${\cal O}(2^n m^2 (n^2 m + n^3)) = {\cal O}(2^n n^2 m^3 + 2^n n^3 m^2)$.

From this analysis, we obtain the following result:

\begin{theorem}
Let  $L=(T,V,E)$ be a link stream  with $|V|=n$ and $|E|=m$, and let $\Delta$ be a duration,
then Algorithm~\ref{alg:dcliques} computes the set of all maximal \dclique s of $L$
in ${\cal O}(2^n n m^2)$ space and ${\cal O}(2^n n^2 m^3 + 2^n n^3 m^2)$ time.
\end{theorem}

\medskip

Notice that enumerating the maximal cliques in a graph $G = (V,E) $ is equivalent to enumerating the maximal \dclique{}s in $L= ([0,0], V, E')$ where $(0,u,v) \in E'$ if and only if $(u,v)\in E$.  The problem of enumerating maximal \dclique{}s in a link stream is therefore at least as difficult as enumerating maximal cliques in a graph, which has an exponential time complexity (in particular, there can be an exponential number of maximal cliques). Therefore any algorithm for enumerating maximal \dclique{}s in a link stream is at least exponential in the number of nodes.

Notice also that several optimizations may speed up our algorithm (without changing its worst-case complexity). In particular, $f$ and $l$, computed in Lines~\ref{alg:get_f} and~\ref{alg:get_l}, are necessarily in $[b,\min(e,b+\Delta)]$ and $[\max(b,e-\Delta),e]$, respectively. One may therefore focus the search on these intervals rather than $[b,e]$. Likewise, if $V(C)$ is the set of nodes satisfying condition of Line~\ref{alg:check_clique_node},
then the set $V(C')$ of nodes satisfying this condition for the \dclique{}s $C'$ added to $S$ at Lines~\ref{alg:add_clique_node}, \ref{alg:add_clique_time_b} and~\ref{alg:add_clique_time_e} is included in $V(C)$. One may therefore associate to each element of $S$ a set of candidate nodes to be considered at Line~\ref{alg:add_node_begin} in place of $V\setminus X$, thus drastically reducing the number of iterations of this loop.

\section{Experiments}
\label{sec:exp}


We implemented Algorithm~\ref{alg:dcliques} with the optimizations discussed above in Python (2.7) and provide the source code at~\cite{dcliquescode2014}.
We illustrate here its practical relevance by computing maximal \dclique s of the link stream
from the {\sc Thiers-Highschool} dataset,
which is a trace of
real-world contacts between individuals, captured with sensors.
It was collected at a French high school in 2012, see~\cite{Fournet2014} for full details. It induces a link stream of $181$ nodes and $45,047$ links,
connecting $2,220$ distinct pairs of nodes over a period of $729,500$ seconds (approximately $8$ days).
Each link $(t,u,v)$ means that the sensor carried by individual $u$ or $v$ detected the sensor carried by the other individual at time $t$, which means in turn that these two individuals were close enough from each other at time $t$ for the detection to happen. We call this a contact between individuals $u$ and $v$.
We also have the information of the class to which students belong.

We computed all maximal $\Delta$-cliques for $\Delta = 60$ seconds, $\Delta = 900$ seconds ($15$ minutes), $\Delta = 3,600$ seconds ($1$ hour), and $\Delta = 10,800$ seconds ($3$ hours).  We handpicked these values because of the rhythm of school day: on a typical day, courses usually last roughly two hours, with two $15$ minutes breaks during the day, and a longer $1$ hour lunch break. Our Python implementation took an hour on a standard server \footnote{A Debian machine with a 2.9 GHz CPU and 64 GB of RAM.} to obtain the results.
Although many discovered \dclique{}s are very small, we also found rather large and long ones.
See Table~\ref{tab:cliques-results} for a summary of these computations.

\begin{table}[tbph]
	\begin{center}
		\begin{tabular}{|P{0.10\linewidth}|P{0.10\linewidth}|P{0.10\linewidth}|P{0.15\linewidth}|P{0.15\linewidth}|P{0.15\linewidth}|}
		\hline
		{\vfill\bf $\Delta$ (s)\vfill} & {\vfill\bf $|R|$\vfill} & {\vfill \bf Max $|X|$ \vfill} & {\vfill\bf Max $e - b$ (s)\vfill} & {\vfill\bf Running time (s)\vfill} & {\vfill\bf Memory (MB)\vfill} \\
		\hline
		\multicolumn{1}{|r|}{60} & \multicolumn{1}{r|}{14 664} & \multicolumn{1}{r|}{5} & \multicolumn{1}{r|}{6 820} & \multicolumn{1}{r|}{150} & \multicolumn{1}{r|}{537} \\
		\multicolumn{1}{|r|}{900} & \multicolumn{1}{r|}{8 214} & \multicolumn{1}{r|}{7} & \multicolumn{1}{r|}{17 420} &\multicolumn{1}{r|}{555} & \multicolumn{1}{r|}{4 755} \\
		\multicolumn{1}{|r|}{3 600} & \multicolumn{1}{r|}{7 170}  & \multicolumn{1}{r|}{7} & \multicolumn{1}{r|}{36 340} & \multicolumn{1}{r|}{1 080} & \multicolumn{1}{r|}{23 186} \\
		\multicolumn{1}{|r|}{10 800} & \multicolumn{1}{r|}{7 416}  & \multicolumn{1}{r|}{7} & \multicolumn{1}{r|}{59  560} & \multicolumn{1}{r|}{3 100} & \multicolumn{1}{r|}{30 453} \\
		\hline
		\end{tabular}
		\caption{Experimental results for computing all maximal $\Delta$-cliques on the {\sc Thiers-Highschool} dataset. $|R|$ is the size of the set returned by our algorithm, (\ie{}, the number of \dclique s found).
For information, storing the dataset in RAM requires 51 MB.}
		\label{tab:cliques-results}	
	\end{center}
\end{table}

\medskip

We present in Figure~\ref{fig:distr}, for each value of $\Delta$, the complementary cumulative distributions for the size $|X|$ and duration $e-b$ of all maximal \dclique s $(X,[b,e])$.
By definition, larger values of $\Delta$ trivially induce larger and longer $\Delta$-cliques. Indeed, if $\Delta'>\Delta$ then every (maximal) \dclique\ also is a $\Delta'$-clique (not maximal in general). More intuitively, small values of $\Delta$ detect local bursts, but are unable to find periodic behaviors if the period is larger than $\Delta$. Notice that when $\Delta$ grows the number of maximal \dclique s generally decreases, but this is not always true, as seen in Table~\ref{tab:cliques-results}.
For an example of how the impact of $\Delta$ on the number of  maximal \dclique{}s is not trivial,
consider the stream presented in Figure~\ref{fig:example}:
it contains four maximal $1$-cliques, 
six maximal $2$-cliques,
and four maximal $3$-cliques.

\begin{figure*}[!htbp]
\centering
\includegraphics[angle=-90,width=0.45\textwidth]{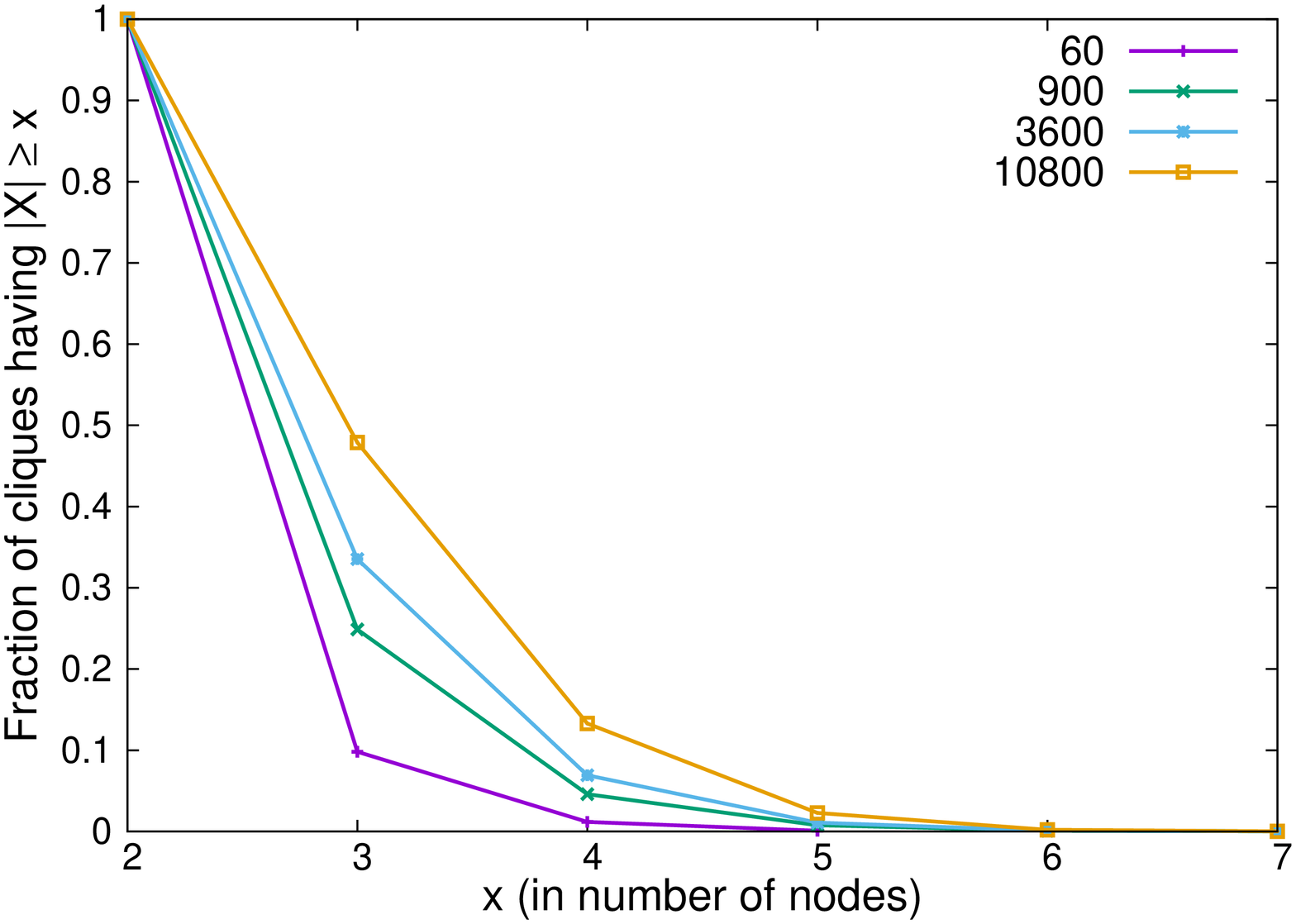}
\includegraphics[angle=-90,width=0.45\textwidth]{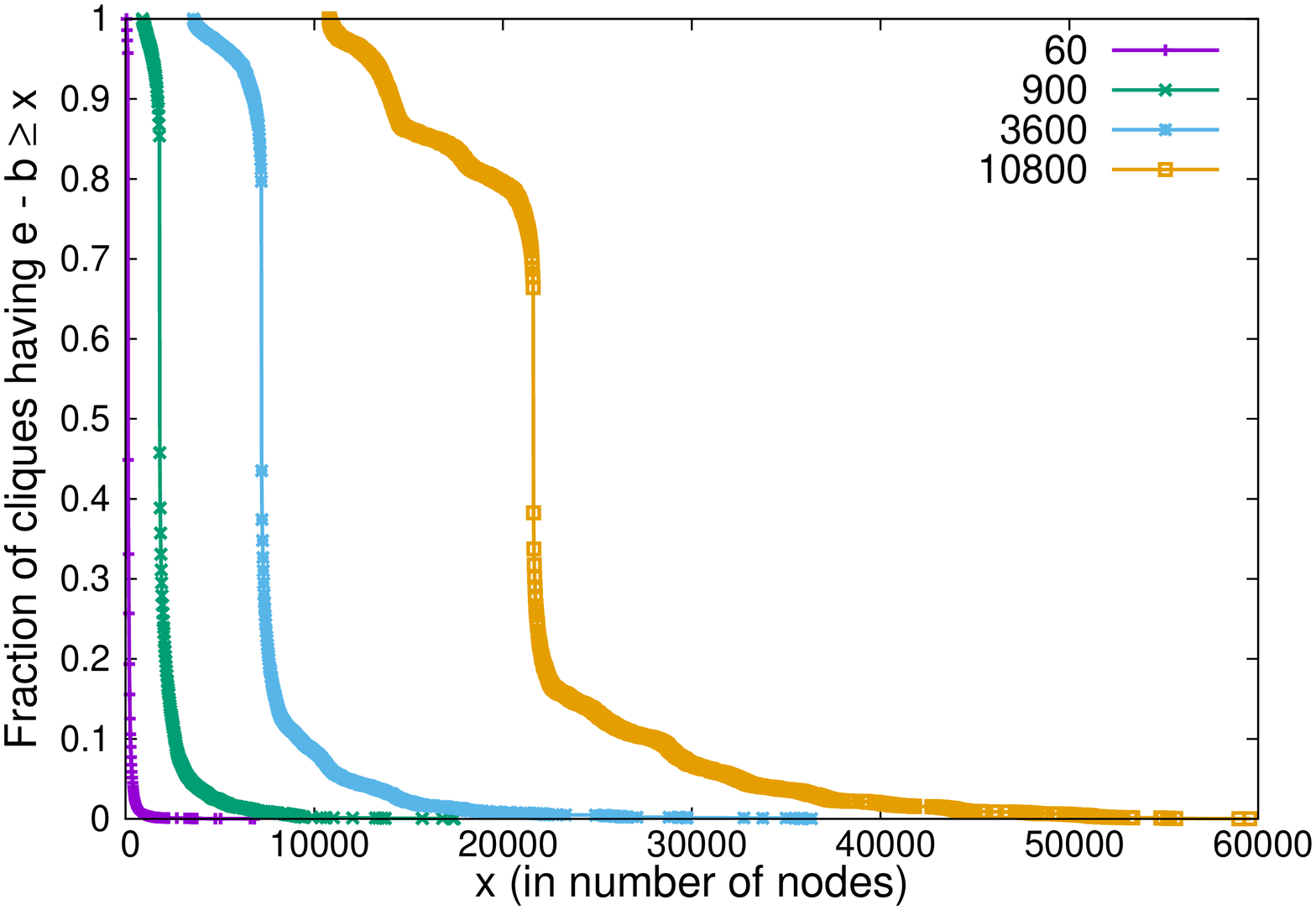}
\caption{{\bf Left:} complementary cumulative distribution of \dclique{} sizes for different values of $\Delta$.
{\bf Right:} complementary cumulative distribution of \dclique{} durations for different values of $\Delta$. The sharp drop at $2\cdot\Delta$ is due to $\Delta$-cliques involving only one link.}
\label{fig:distr}
\end{figure*}


\medskip

Notice now that Algorithm~\ref{alg:dcliques} makes no assumption on the order in which elements of $S$ are processed, which corresponds to the way we explore the configuration space. In particular, if $S$ is a first-in-first-out structure (a queue), the algorithm performs a BFS of the configuration space; if it is a last-in-first-out structure (a stack) then it performs a DFS. The execution time is essentially the same in all cases. The size of $S$ may vary, but the space complexity of the algorithm is dominated by the size of $M$, that does not change. Still, the data structure impacts the order in which \dclique{}s are found.

\begin{figure}[!h]
\centering
\includegraphics[angle=-90,width=0.49\linewidth]{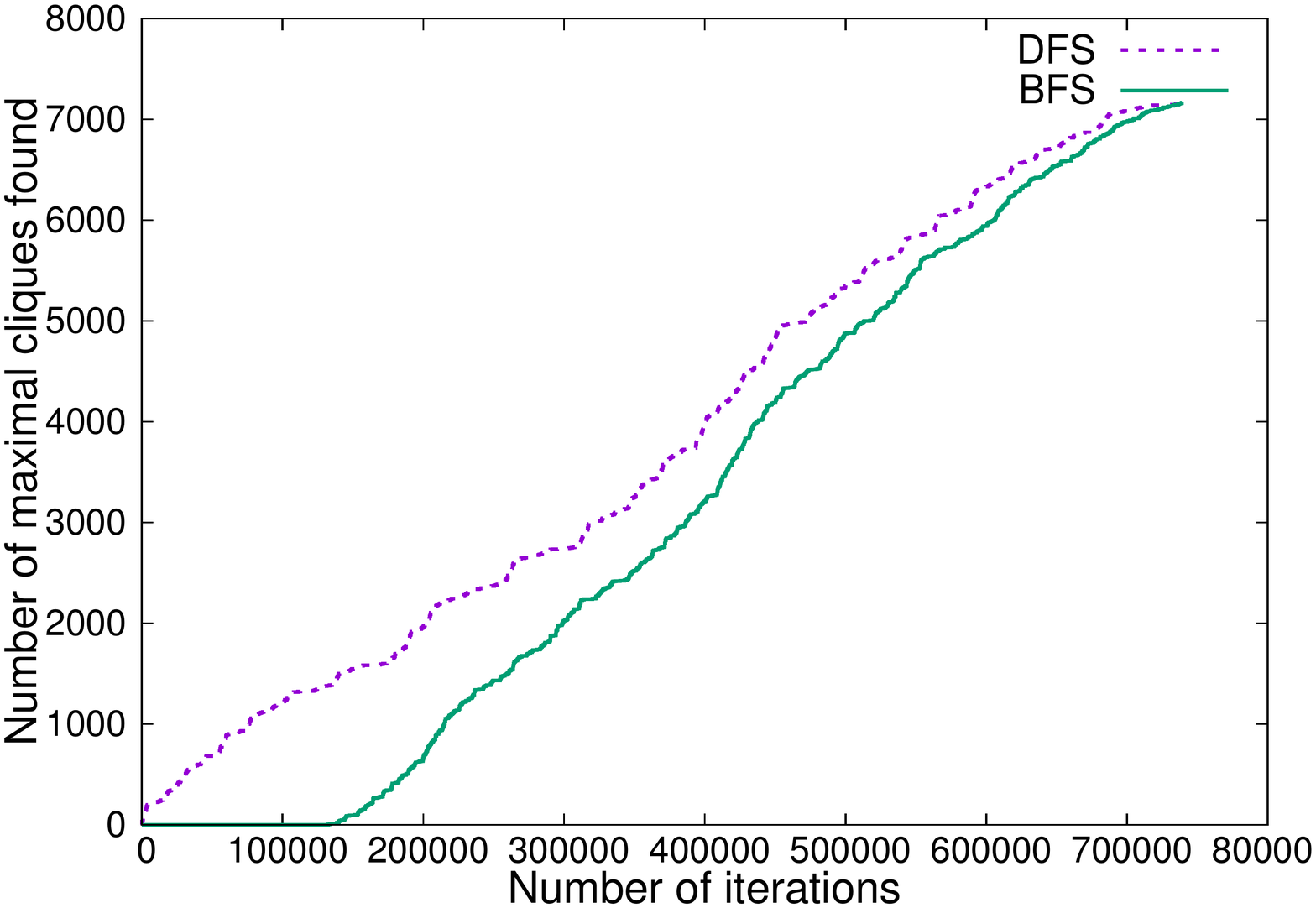}
\includegraphics[angle=-90,width=0.49\linewidth]{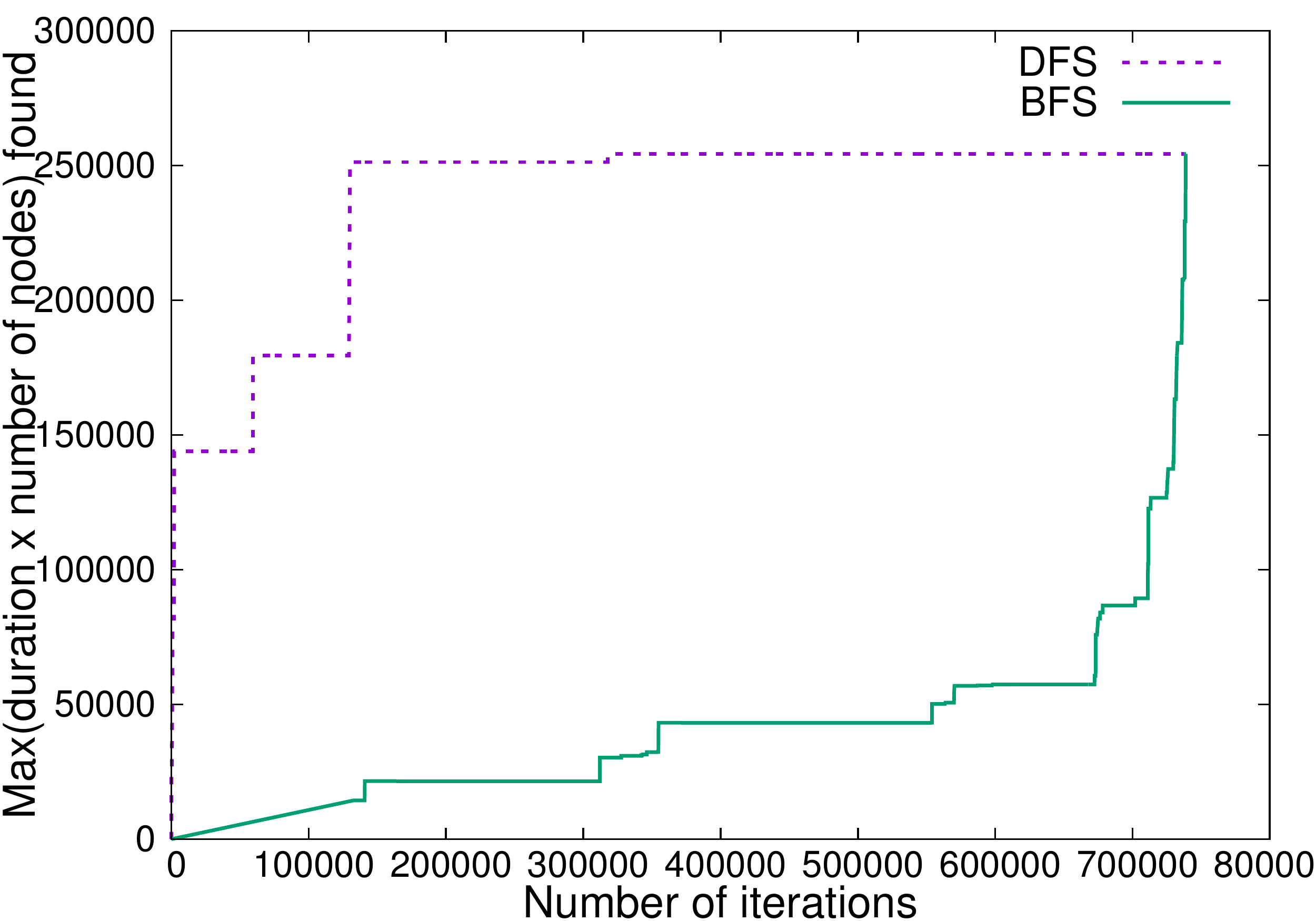}
\caption{Behavior of our algorithm depending on the way it explores its configuration space (DFS or BFS). {\bf Left}: number of maximal cliques discovered as a function of the number of iterations of the main loop of the algorithm. {\bf Right}: maximal size of discovered cliques as a function of the number of iterations of the main loop of the algorithm. A clique size is estimated here by its number of nodes times its duration (in seconds).}
\label{fig:lifo-fifo}
\end{figure}

We illustrate this in the practical case where $\Delta =$ 3600 seconds (1 \mbox{ hour}), see Figure~\ref{fig:lifo-fifo}. It shows that DFS rapidly discovers many cliques, and that those cliques are non-trivial cliques (cliques involving more than $2$ nodes or lasting a substantial amount of time). In this case, using a DFS is therefore more interesting than a BFS, as it outputs results and exhibits non-trivial \dclique{}s faster. However, this behavior is dependent on the dataset, and deciding on the most appropriate exploration strategy in a given case remains an open question.

\medskip

Consider now $\mathcal{G} = (\mathcal{V}, \mathcal{E})$ the graph induced by link stream $L = (T,V,E)$: $\mathcal{V} = V$ and $\mathcal{E} = \{(u,v):\exists (t,u,v)\in E\}$. In other words, this is the graph where a link exists between two nodes $u$ and $v$ if and only if there is at least one contact between $u$ and $v$ in the link stream. The graph $\mathcal{G}$ contains $1742$ cliques, the largest one involving $14$ nodes.
Approximately 70\% of them involve students in the same class.

If $(X,[b,e])$ is a (maximal) \dclique\ of $L$, then by definition $X$ is a clique of $\mathcal{G}$ (in general not maximal). However, as \dclique s capture time information they shed light on different patterns. For instance, $L$ contains a $60$-clique involving $4$ students of different classrooms during roughly $5$ minutes, which is likely to be the signature of a coffee break. Such \dclique s are non-trivial outputs of our algorithm, but they are invisible when considering graph cliques.

\section{Conclusion}

We introduced the notion of \dclique s in link streams, and proposed the first algorithm to compute the maximal such \dclique{}s. We implemented this algorithm and detected interesting \dclique s in real-world data.

Clearly, our algorithm may be improved further. Trying to adapt the Bron-Kerbosch algorithm~\cite{Bron1973} and some of its variants~\cite{Tomita2006,Koch2001, Cazals2008, Eppstein2013}, the most widely used algorithms for computing cliques in graphs, is particularly appealing. Indeed, the configuration spaces built by these algorithms are trees,
which avoids redundant computations.
This is achieved by maintaining a set of candidate nodes that may be added to previously discovered cliques, which does not directly translate to our situation because of time in link streams. Still, we believe that progress is possible in this direction.

We also consider the case of links with duration as a promising perspective: each link $(b,e,u,v)$ means that $u$ and $v$ interact from time $b$ to $e$. In this case there is no need for a $\Delta$ anymore, as density in this context is nothing but the probability that two randomly chosen nodes are linked together at a randomly chosen time. The definition of cliques in link streams with durations follows directly, and our algorithm may be extended to compute maximal such cliques.

\medskip

{\small\noindent


{\bf Acknowledgments.}

We warmly thank the anonymous reviewers, who helped us improve this paper much. This work is supported in part by the French {\em Direction Générale de l’Armement} (DGA), by the Thales company, by the CODDDE ANR-13-CORD-0017-01 grant from the {\em Agence Nationale de la Recherche}, and by grant O18062-44430 of the French program {\em PIA -- Usages, services et contenus innovants}.

}

\section*{References}

\bibliographystyle{abbr}

\bibliography{dcliques-algorithm}

\begin{thebibliography}{10}

\bibitem{Bron1973}
C.~Bron and J.~Kerbosch.
\newblock {Algorithm 457: finding all cliques of an undirected graph}.
\newblock {\em Communications of the ACM}, 16(9), 1973.

\bibitem{Casteigts2011}
A.~Casteigts, P.~Flocchini, W.~Quattrociocchi, and N.~Santoro.
\newblock Time-varying graphs and dynamic networks.
\newblock In {\em Ad-hoc, Mobile, and Wireless Networks}, volume 6811 of {\em
  Lecture Notes in Computer Science}, pages 346--359. 2011.

\bibitem{Cazals2008}
F.~Cazals and C.~Karande.
\newblock A note on the problem of reporting maximal cliques.
\newblock {\em Theoretical Computer Science}, 407(1–3):564 -- 568, 2008.

\bibitem{Eppstein2013}
D.~Eppstein, M.~L\"{o}ffler, and D.~Strash.
\newblock Listing all maximal cliques in large sparse real-world graphs.
\newblock {\em Journal of Experimental Algorithmics}, 18:3.1:3.1--3.1:3.21,
  2013.

\bibitem{Fournet2014}
J.~Fournet and A.~Barrat.
\newblock Contact patterns among high school students.
\newblock {\em PLoS ONE}, 9:e107878, 2014.

\bibitem{Holme2011}
P.~Holme and J.~Saramäki.
\newblock Temporal networks.
\newblock {\em Physics Reports}, 519:97--125, 2012.

\bibitem{Johnson1988}
D.~S. Johnson, M.~Yannakakis, and C.~H. Papadimitriou.
\newblock On generating all maximal independent sets.
\newblock {\em Information Processing Letters}, 27(3):119--123, 3 1988.

\bibitem{Koch2001}
I.~Koch.
\newblock {Enumerating all connected maximal common subgraphs in two graphs}.
\newblock {\em Theoretical Computer Science}, 250:1--30, 2001.

\bibitem{Rowe2007}
R.~Rowe, G.~Creamer, S.~Hershkop, and S.~J. Stolfo.
\newblock Automated social hierarchy detection through email network analysis.
\newblock In {\em Proceedings of the 9th WebKDD and 1st SNA-KDD 2007 Workshop
  on Web Mining and Social Network Analysis}, WebKDD/SNA-KDD '07, pages
  109--117, New York, NY, USA, 2007. ACM.

\bibitem{Samudrala1998}
R.~Samudrala and J.~Moult.
\newblock A graph-theoretic algorithm for comparative modeling of protein
  structure.
\newblock {\em Journal of Molecular Biology}, 279(1):287 -- 302, 1998.

\bibitem{Tomita2006}
E.~Tomita, A.~Tanaka, and H.~Takahashi.
\newblock {The worst-case time complexity for generating all maximal cliques
  and computational experiments}.
\newblock {\em Theoretical Computer Science}, 363:28--42, 2006.

\bibitem{Viard2014}
T.~Viard and M.~Latapy.
\newblock Identifying roles in an {I}{P} network with temporal and structural
  density.
\newblock In {\em Computer Communications Workshops (INFOCOM WKSHPS)}, pages
  801--806, 2014.

\bibitem{dcliquescode2014}
T.~Viard and M.~Latapy.
\newblock Source code in python for computing cliques in link streams:
  \texttt{https://github.com/TiphaineV/delta-cliques}, 2014.

\bibitem{Wehmuth2014}
K.~Wehmuth, A.~Ziviani, and E.~Fleury.
\newblock {A Unifying Model for Representing Time-Varying Graphs}.
\newblock Research Report RR-8466, ENS Lyon, 2014.

\end{thebibliography}

\end{document}